\documentclass{llncs}

\usepackage{latexsym,amsmath,amssymb}
\usepackage{graphicx,epsfig}
\usepackage{url}
\usepackage{multirow}
\usepackage[nocompress]{cite}
 
\newcommand{\Suff}{\textit{Suff}}
\newcommand{\Pref}{\textit{Pref}}
\newcommand{\Fact}{\textit{Fact}}

\newcommand{\St}{\textit{St}}
\newcommand{\LS}{\textit{LS}}
\newcommand{\RS}{\textit{RS}}
\newcommand{\BS}{\textit{BS}}
\newcommand{\SBS}{\textit{SBS}}
\newcommand{\WBS}{\textit{NBS}}
\newcommand{\MF}{\textit{MFSt}}

\renewcommand{\epsilon}{\varepsilon}

\begin{document}

\title{A Characterization of Bispecial \\ Sturmian Words}

\author{Gabriele Fici}

\authorrunning{G. Fici}

\institute{I3S, CNRS \& Universit\'e Nice Sophia Antipolis, France \\ \email{fici@i3s.unice.fr} }

\maketitle

\begin{abstract}
A finite Sturmian word $w$ over the alphabet $\{a,b\}$ is left special (resp.~right special)  if $aw$ and $bw$ (resp.~$wa$ and $wb$) are both Sturmian words. A bispecial Sturmian word is a Sturmian word that is both left and right special. We show as a main result that bispecial Sturmian words are exactly the maximal internal factors of Christoffel words, that are words coding the digital approximations of segments in the Euclidean plane. This result is an extension of the known relation between central words and primitive Christoffel words. Our characterization allows us to give an enumerative formula for bispecial Sturmian words. We also investigate the minimal forbidden words for the set of Sturmian words.
\end{abstract}


\section{Introduction}\label{sec:intro}

Sturmian words are non-periodic infinite words of minimal factor complexity. They are characterized by the property of having exactly $n+1$ distinct factors of length $n$ for every $n\ge 0$ (and therefore are binary words) \cite{MoHe40}. As an immediate consequence of this property, one has that in any Sturmian word there is a unique factor for each length $n$ that can be extended to the right with both letters into a factor of length $n+1$. These factors are called \emph{right special factors}. Moreover, since any Sturmian word is recurrent (every factor appears infinitely often) there is a unique factor for each length $n$ that is left special, i.e., can be extended to the left with both letters into a factor of length $n+1$.

The set $\St$ of finite factors of Sturmian words coincides with the set of binary \emph{balanced} words, i.e., binary words having the property that any two factors of the same length have the same number of occurrences of each letter up to one. These words are also called (finite) Sturmian words and have been extensively studied because of their relevant role in several fields of theoretical computer science.

If one considers extendibility within the set $\St$, one can define \emph{left special Sturmian words} (resp.~\emph{right special Sturmian words}) \cite{DelMi94} as those words $w$ over the alphabet $\Sigma=\{a,b\}$ such that $aw$ and $bw$ (resp.~$wa$ and $wb$) are both Sturmian words. For example, the word $aab$ is left special since $aaab$ and $baab$ are both Sturmian words, but is not right special since $aabb$ is not a Sturmian word.

Left special Sturmian words are precisely the binary words having suffix automaton\footnote{The suffix automaton of a finite word $w$ is the minimal deterministic finite state automaton accepting the language of the suffixes of $w$.} with minimal state complexity (cf.~\cite{SciZa07,Fi10b}). From combinatorial considerations one has that right special Sturmian words are the reversals of left special Sturmian words.

The Sturmian words that are both left special and right special are called \emph{bispecial Sturmian words}. They are of two kinds: \emph{strictly bispecial Sturmian words}, that are the words $w$ such that $awa$, $awb$, $bwa$ and $bwb$ are all Sturmian words, or \emph{non-strictly bispecial Sturmian words} otherwise. Strictly bispecial Sturmian words have been deeply studied (see for example~\cite{DelMi94,CarDel05}) because they play a central role in the theory of Sturmian words. They are also called \emph{central words}. Non-strictly bispecial Sturmian words, instead, received less attention.

One important field in which Sturmian words arise naturally is discrete geometry. Indeed, Sturmian words can be viewed as digital approximations of straight lines in the Euclidean plane. It is known that given a point $(p,q)$ in the discrete plane $\mathbb{Z} \times \mathbb{Z}$, with $p,q>0$, there exists a unique path that approximates from below (resp.~from above) the segment joining the origin $(0,0)$ to the point $(p,q)$. This path, represented as a concatenation of horizontal and vertical unitary segments, is called the \emph{lower (resp.~upper) Christoffel word} associated to the pair $(p,q)$. If one encodes horizontal and vertical unitary segments with the letters $a$ and $b$ respectively, a lower (resp.~upper) Christoffel word is always a word of the form $awb$ (resp.~$bwa$), for some $w\in \Sigma^{*}$. If (and only if) $p$ and $q$ are coprime, the associated Christoffel word is primitive (that is, it is not the power of a shorter word). It is known that a word $w$ is a strictly bispecial Sturmian word if and only if $awb$ is a primitive lower Christoffel word (or, equivalently, if and only if $bwa$ is a primitive upper Christoffel word). As a main result of this paper, we show that this correspondence holds in general between bispecial Sturmian words and Christoffel words. That is, we prove (in Theorem \ref{theor:main}) that $w$ is a bispecial Sturmian word if and only if there exist letters $x,y$ in $\{a,b\}$ such that $xwy$ is a Christoffel word. 

This characterization allows us to prove an enumerative formula for bispecial Sturmian words (Corollary \ref{cor:formula}): there are exactly $2n+2-\phi(n+2)$ bispecial Sturmian words of length $n$, where $\phi$ is the Euler totient function, i.e., $\phi(n)$ is the number of positive integers smaller than or equal to $n$ and coprime with $n$. It is worth noticing that enumerative formulae for left special, right special and strictly bispecial Sturmian words were known \cite{DelMi94}, but to the best of our knowledge we exhibit the first proof of an enumerative formula for non-strictly bispecial (and therefore for bispecial) Sturmian words.

We then investigate minimal forbidden words for the set $\St$ of finite Sturmian words. The set of \emph{minimal forbidden words} of a factorial language $L$ is the set of words of minimal length that do not belong to $L$ \cite{MiReSci02}. Minimal forbidden words represent a powerful tool to investigate the structure of a factorial language (see~\cite{BeMiRe96}).
We give a characterization of minimal forbidden words for the set of Sturmian words in Theorem \ref{theor:mf}. We show that they are the words of the form $ywx$ such that $xwy$ is a non-primitive Christoffel word, where $\{x,y\}=\{a,b\}$. This characterization allows us to give an enumerative formula for the set of minimal forbidden words (Corollary \ref{cor:formulamf}):  there are exactly $2(n-1-\phi(n))$ minimal forbidden words of length $n$ for every $n>1$.

The paper is organized as follows. In Sec.\ \ref{sec:wsf} we recall standard definitions on words and factors. In Sec.\ \ref{sec:StCh} we deal with Sturmian words and Christoffel words, and present our main result. In Sec.\ \ref{sec:En} we give an enumerative formula for bispecial Sturmian words. Finally, in Sec.\ \ref{sec:MF}, we investigate minimal forbidden words for the language of finite Sturmian words. 

\section{Words and special factors}\label{sec:wsf}

Let $\Sigma$ be a finite alphabet, whose elements are called letters. A word over $\Sigma$ is a finite sequence of letters from $\Sigma$. A right-infinite word over $\Sigma$ is a non-ending sequence of letters from $\Sigma$. The set of all words over $\Sigma$ is denoted by $\Sigma^*$. The set of all words over $\Sigma$ having length $n$ is denoted by $\Sigma^n$. The empty word has length zero and is denoted by $\varepsilon$. For a subset $X$ of $\Sigma^{*}$, we note $X(n)=|X\cap \Sigma^{n}|$. Given a non-empty word $w$, we let $w[i]$ denote its $i$-th letter. The reversal of the word $w=w[1]w[2]\cdots w[n]$, with $w[i]\in\Sigma$ for $1\leq i\leq n$, is the word $\tilde{w}=w[n]w[n-1]\cdots w[1]$. We set $\tilde{\epsilon}=\epsilon$. A palindrome is a word $w$ such that $\tilde{w}=w$. A word is called a power if it is the concatenation of copies of another word; otherwise it is called primitive. For a letter $a\in \Sigma$, $|w|_{a}$ is the number of $a$'s appearing in $w$.  A word $w$ has period $p$, with $0<p\leq |w|$, if $w[i]=w[i+p]$ for every $i=1,\ldots ,|w|-p$. Since $|w|$ is always a period of $w$, every non-empty word has at least one period. %

A word $z$ is a factor of a word $w$ if $w=uzv$ for some $u,v\in \Sigma^{*}$. In the special case $u = \varepsilon $ (resp.~$v = \varepsilon $), we call $z$ a prefix (resp.~a suffix) of $w$. We let $\Pref(w)$, $\Suff(w)$ and $\Fact(w)$ denote, respectively, the set of prefixes, suffixes and factors of the word $w$. The factor complexity of a word $w$ is the integer function $f_{w}(n)=|\Fact(w)\cap \Sigma^n|$,  $n\geq 0$. 

A factor $u$ of a word $w$ is left special (resp.~right special) in $w$ if there exist $a,b\in \Sigma$, $a\neq b$, such that $au,bu\in \Fact(w)$ (resp.~$ua,ub\in \Fact(w)$). A factor $u$ of $w$ is bispecial in $w$ if it is both left and right special. In the case when $\Sigma=\{a,b\}$, a bispecial factor $u$ of $w$ is said to be strictly bispecial in $w$ if $aua,aub,bua,bub$ are all factors of $w$; otherwise $u$ is said to be non-strictly bispecial in $w$. For example, let $w= aababba$. The left special factors of $w$ are $\epsilon$, $a$, $ab$, $b$ and $ba$. The right special factors of $w$ are $\epsilon$, $a$, $ab$ and $b$. Therefore, the bispecial factors of $w$ are $\epsilon$, $a$, $ab$ and $b$. Among these, only $\epsilon$ is strictly bispecial.

In the rest of the paper we fix the alphabet $\Sigma=\{a,b\}$.

\section{Sturmian words and Christoffel words}\label{sec:StCh}

A right-infinite word $w$ is called a Sturmian word if $f_{w}(n)=n+1$ for every $n\ge 0$, that is, if $w$ contains exactly $n+1$ distinct factors of length $n$ for every $n\ge 0$. Sturmian words are non-periodic infinite words of minimal factor complexity \cite{CoHe73}.
A famous example of infinite Sturmian word is the Fibonacci word $$F=abaababaabaababaababaabaababaabaab\cdots$$ obtained as the limit of the substitution $a\mapsto ab$, $b\mapsto a$. 

A finite word is called Sturmian if it is a factor of an infinite Sturmian word. Finite Sturmian words are characterized by the following balance property \cite{DuGB90}: a finite word $w$ over $\Sigma=\{a,b\}$ is Sturmian if and only if for any $u,v\in \Fact(w)$ such that $|u|=|v|$ one has $||u|_{a}-|v|_{a}|\le 1$ (or, equivalently, $||u|_{b}-|v|_{b}|\le 1$). We let $\St$ denote the set of finite Sturmian words. The language $\St$ is factorial (if $w=uv\in \St$, then $u,v\in \St$) and extendible (for every $w\in \St$ there exist letters $x,y\in \Sigma$ such that $xwy\in \St$).

Let $w$ be a finite Sturmian word. The following definitions are in \cite{DelMi94}. 

\begin{definition}
A word  $w\in \Sigma^{*}$ is a  left special (resp.~right special) Sturmian word if $aw,bw\in \St$ (resp.~if $wa,wb\in \St$). A bispecial Sturmian word is a Sturmian word that is both left special and right special. Moreover, a bispecial Sturmian word is strictly bispecial if $awa,awb,bwa$ and $bwb$ are all Sturmian word; otherwise it is non-strictly bispecial. 
\end{definition}

We let $\LS$, $\RS$, $\BS$, $\SBS$ and $\WBS$ denote, respectively, the sets of left special, right special, bispecial, strictly bispecial and non-strictly bispecial Sturmian words. Hence, $\BS=\LS\cap \RS=\SBS\cup \WBS$.

The following lemma is a reformulation of a result of de Luca \cite{Del97}.

\begin{lemma}\label{lem:prefsuf}
Let $w$ be a word over $\Sigma$. Then $w\in \LS$ (resp.~$w\in \RS$) if and only if $w$ is a prefix (resp.~a suffix) of a word in $\SBS$.
\end{lemma}

Given a bispecial Sturmian word, the simplest criterion to determine if it is strictly or non-strictly bispecial  is provided by the following nice characterization \cite{DelMi94}:

\begin{proposition}\label{prop:sturmstrispe}
 A bispecial Sturmian word is strictly bispecial if and only if it is a palindrome.
\end{proposition}

Using the results in \cite{DelMi94}, one can derive the following classification of Sturmian words with respect to their extendibility.

\begin{proposition}\label{prop:bisp}
 Let $w$ be a Sturmian word. Then:
 
\begin{itemize}
\item  $|\Sigma w\Sigma\cap \St|=4$ if and only if $w$ is strictly bispecial;
\item  $|\Sigma w\Sigma\cap \St|=3$ if and only if $w$ is non-strictly bispecial;
\item  $|\Sigma w\Sigma\cap \St|=2$ if and only if $w$ is left special or right special but not bispecial;
\item  $|\Sigma w\Sigma\cap \St|=1$ if and only if $w$ is neither left special nor right special.
\end{itemize}
\end{proposition}

\begin{example}
The word $w=aba$ is a strictly bispecial Sturmian word, since $awa$, $awb$, $bwa$ and $bwb$ are all Sturmian words, so that $|\Sigma w\Sigma\cap \St|=4$. The word $w=abab$ is a bispecial Sturmian word since $wa$, $wb$, $aw$ and $bw$ are Sturmian words. Nevertheless, $awb$ is not Sturmian, since it contains $aa$ and $bb$ as factors. So $w$ is a non-strictly bispecial Sturmian word, and $|\Sigma w\Sigma\cap \St|=3$. The Sturmian word $w=aab$ is left special but not right special, and $|\Sigma w\Sigma\cap \St|=2$. Finally, the Sturmian word $w=baab$ is neither left special nor right special, the only word in $\Sigma w\Sigma\cap \St$ being $awa$.
\end{example}

We now recall the definition of central word \cite{DelMi94}.

\begin{definition}
A word over $\Sigma$ is central if it has two coprime periods $p$ and $q$ and length equal to $p+q-2$.
\end{definition}

A combinatorial characterization of central words is the following (see~\cite{Del97}):

\begin{proposition}\label{prop:charcentral}
A word $w$ over $\Sigma$ is central if and only if $w$ is the power of a single letter or there exist palindromes $P,Q$ such that $w=PxyQ=QyxP$, for different letters $x,y\in \Sigma$. Moreover, if $|P|<|Q|$, then $Q$ is the longest palindromic suffix of $w$.
\end{proposition}

Actually, in the statement of Proposition \ref{prop:charcentral}, the requirement that the words $P$ and $Q$  are palindromes is not even necessary \cite{CarDel05}.

We have the following remarkable result \cite{DelMi94}:

\begin{proposition}\label{prop:sbscen}
A word over $\Sigma$ is a strictly bispecial Sturmian word if and only if it is a central word.
\end{proposition}

Another class of finite words, strictly related to the previous ones, is that of Christoffel words.

\begin{definition}
Let $n>1$ and $p,q>0$ be integers such that $p+q=n$. The lower Christoffel word $w_{p,q}$ is the word defined for $1\le i\le n$ by
\[w_{p,q}[i] = \left\{ \begin{array}{lllll}
a & \mbox{if $iq$ mod$(n)>(i-1)q$ mod$(n)$,}\\
b & \mbox{if $iq$ mod$(n)<(i-1)q$ mod$(n)$.}\\
\end{array} \right.\]
\end{definition}

\begin{example}
 Let $p=6$ and $q=4$. We have $\{i4$ mod$(10)\mid i=0,1,\ldots,10\}=\{0,4,8,2,6,0,4,8,2,6,0\}$. Hence, $w_{6,4}=aababaabab$.
\end{example}

Notice that for every $n>1$, there are exactly $n-1$ lower Christoffel words $w_{p,q}$, corresponding to the $n-1$ pairs $(p,q)$ such that $p,q>0$ and $p+q=n$. 

\begin{remark}
In the literature, Christoffel words are often defined with the additional requirement that $\gcd(p,q)=1$ (cf.~\cite{Book08}). We call such Christoffel words primitive, since a Christoffel word is a primitive word if and only if $\gcd(p,q)=1$.
\end{remark}

If one draws a word in the discrete grid $\mathbb{Z} \times \mathbb{Z}$ by encoding each $a$ with a horizontal unitary segment and each $b$ with a vertical unitary segment, the lower Christoffel word $w_{p,q}$ is in fact the best grid approximation from below of the segment joining $(0,0)$ to $(p,q)$, and has slope $q/p$, that is, $|w|_{a}=p$ and  $|w|_{b}=q$ (see~\figurename~\ref{fig:GC}).

\begin{figure}
\begin{center}
\begin{minipage}{5.7cm}
\includegraphics[height=40mm]{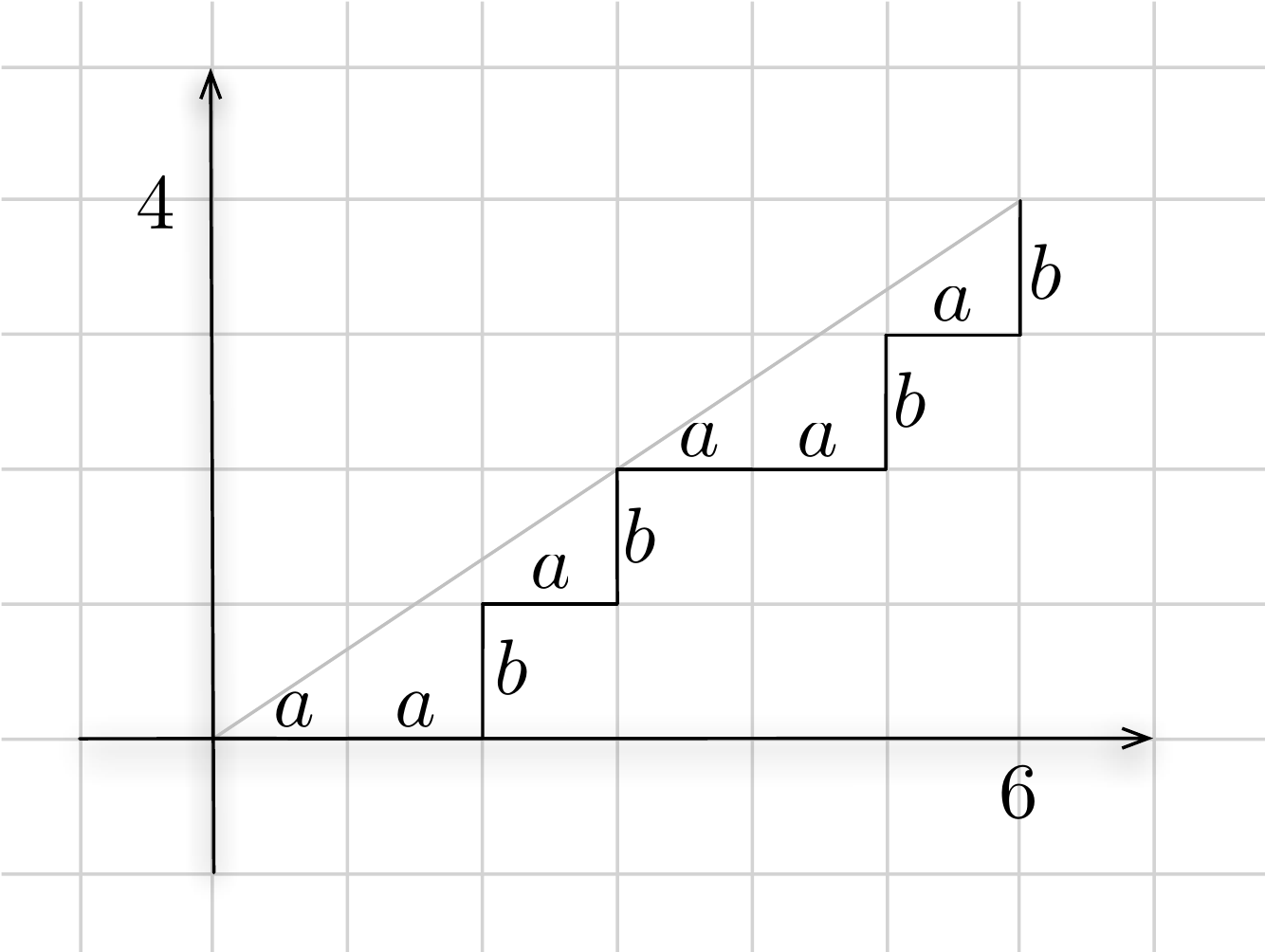}
\end{minipage}
\begin{minipage}{5.7cm}
\includegraphics[height=40mm]{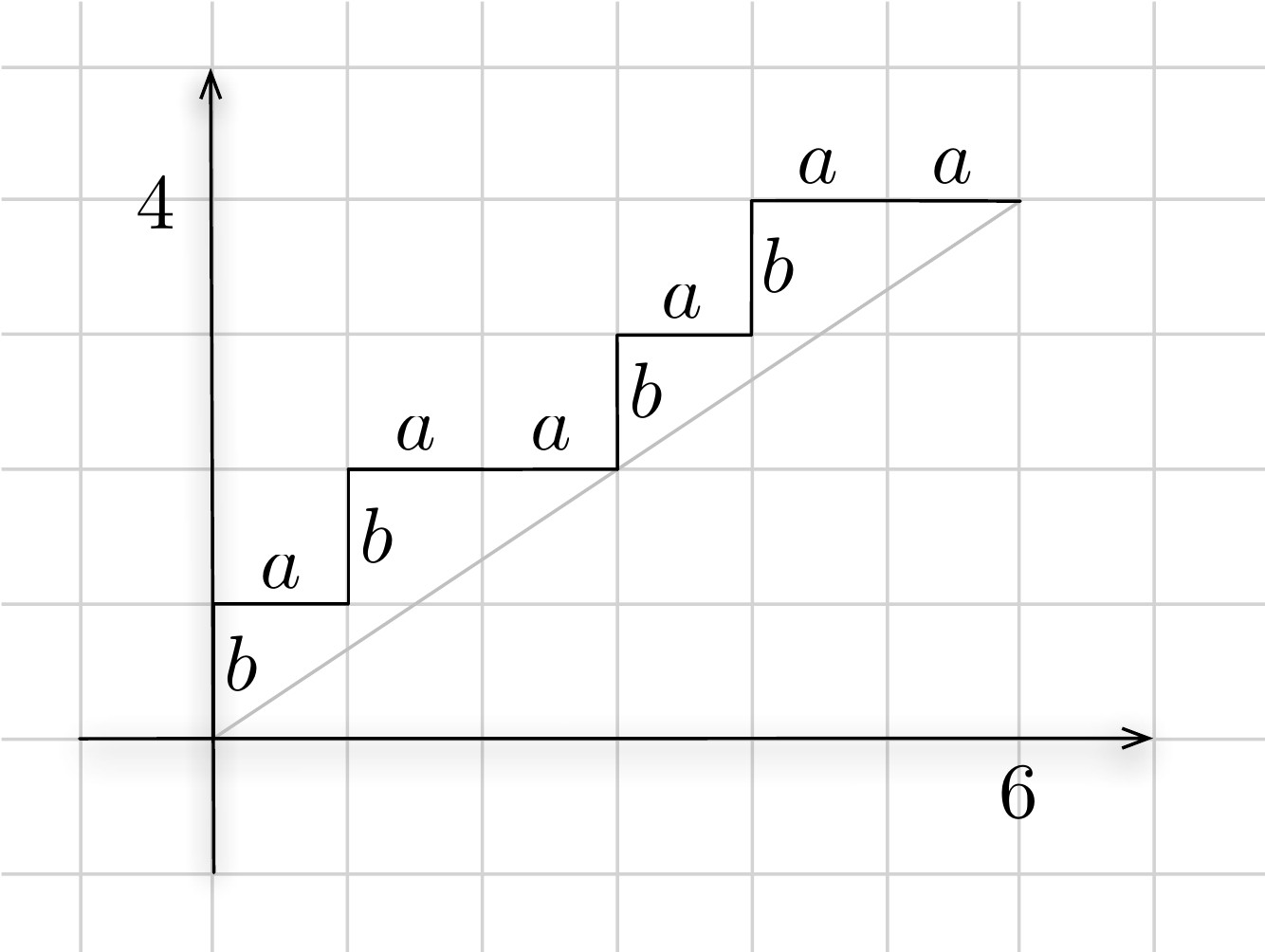}
\end{minipage}
\end{center}
\caption{The lower Christoffel word $w_{6,4}=aababaabab$ (left) and the upper Christoffel word $w'_{6,4}=babaababaa$ (right).\label{fig:GC}}
\end{figure}

Analogously, one can define the upper Christoffel word $w'_{p,q}$ by
\[w'_{p,q}[i] = \left\{ \begin{array}{lllll}
a & \mbox{if $ip$ mod$(n)<(i-1)p$ mod$(n)$,}\\
b & \mbox{if $ip$ mod$(n)>(i-1)p$ mod$(n)$.}\\
\end{array} \right.\]
Of course, the upper Christoffel word $w'_{p,q}$ is the best grid approximation from above of the segment joining $(0,0)$ to $(p,q)$ (see \figurename~\ref{fig:GC}).

\begin{example}
 Let $p=6$ and $q=4$. We have $\{i6$ mod$(10)\mid i=0,1,\ldots,10\}=\{0,6,2,8,4,0,6,2,8,4,0\}$. Hence, $w'_{6,4}=babaababaa$.
\end{example}

The next result follows from elementary geometrical considerations.

\begin{lemma}\label{lem:rev}
 For every pair $(p,q)$ the word $w'_{p,q}$ is the reversal of the word $w_{p,q}$.
\end{lemma}

If (and only if) $p$ and $q$ are coprime, the Christoffel word $w_{p,q}$ intersects the segment joining $(0,0)$ to $(p,q)$ only at the end points, and is a primitive word. Moreover, one can prove that $w_{p,q}=aub$ and $w'_{p,q}=bua$ for a palindrome $u$. Since $u$ is a bispecial Sturmian word and it is a palindrome,  $u$ is a strictly bispecial Sturmian word (by Proposition \ref{prop:sturmstrispe}). Conversely, given a strictly bispecial Sturmian word $u$, $u$ is a central word (by Proposition \ref{prop:sbscen}), and therefore has two coprime periods $p,q$ and length equal to $p+q-2$. Indeed, it can be proved that $aub=w_{p,q}$ and $bua=w'_{p,q}$. The previous properties can be summarized in the following theorem (cf.~\cite{BeDel97}):

\begin{theorem}\label{theor:sbsCP}
  $\SBS=\{w\mid xwy \mbox{ is a primitive Christoffel word, $x,y\in \Sigma\}$}.$
\end{theorem}

If instead $p$ and $q$ are not coprime, then there exist coprime integers $p',q'$ such that $p=rp'$, $q=rq'$, for an integer $r>1$. In this case, we have $w_{p,q}=(w_{p',q'})^{r}$, that is, $w_{p,q}$ is a power of a primitive Christoffel word. Hence, there exists a central Sturmian word $u$ such that $w_{p,q}=(aub)^{r}$ and $w'_{p,q}=(bua)^{r}$. So, we have:

\begin{lemma}\label{lem:npChris}
The word $xwy$, $x\neq y\in \Sigma$, is a Christoffel word if and only if $w=(uyx)^{n}u$, for an integer $n\ge 0$ and a central word $u$. Moreover, $xwy$ is a primitive Christoffel word if and only if $n=0$.
\end{lemma}

Recall from \cite{Del97} that the right (resp.~left) palindromic closure of a word $w$ is the (unique) shortest palindrome $w^{(+)}$ (resp.~$w^{(-)}$) such that $w$ is a prefix of $w^{(+)}$ (resp.~a suffix of $w^{(-)}$). If $w=uv$ and $v$ is the longest palindromic suffix of $w$ (resp.~$u$ is the longest palindromic prefix of $w$), then $w^{(+)}=w\tilde{u}$ (resp.~$w^{(-)}=\tilde{v}w$). 

\begin{lemma}\label{lem:rpl}
 Let $xwy$ be a Christoffel word, $x,y\in \Sigma$. Then  $w^{(+)}$ and  $w^{(-)}$ are central words.
\end{lemma}

\begin{proof}
 Let $xwy$ be a Christoffel word, $x,y\in \Sigma$.  By Lemma \ref{lem:npChris}, $w=(uyx)^{n}u$, for an integer $n\ge 0$ and a central word $u$. We prove the claim for the right palindromic closure, the claim for the left palindromic closure will follow by symmetry. If $n=0$, then $w=u$, so $w$ is a palindrome and then $w^{(+)}=w$ is a central word. So suppose $n>0$. We first consider the case when $u$ is the power of a single letter (including the case $u=\epsilon$). We have that either $w=(y^{k+1}x)^{n}y^{k}$ or $w=(x^{k}yx)^{n}x^{k}$ for some $k\ge 0$. In the first case, $w^{(+)}=wy=(y^{k+1}x)^{n}y^{k+1}$, whereas in the second case $w^{(+)}=wyx^{k}=(x^{k}yx)^{n}x^{k}yx^{k}$. In both cases one has that $w^{(+)}$ is a strictly bispecial Sturmian word, and thus, by Proposition \ref{prop:sbscen}, a central word.

Let now $u$ be not the power of a single letter. Hence, by Proposition \ref{prop:charcentral}, there exist palindromes $P,Q$ such that $u=PxyQ=QyxP$. Now, observe that 
\[
 w  =  (uyx)^{n}u
   =  Pxy(QyxPxy)^{n}Q
\]
We claim that the longest palindromic suffix of $w$ is $(QyxPxy)^{n}Q$. Indeed, the longest palindromic suffix of $w$ cannot be $w$ itself since $w$ is not a palindrome, so since any  palindromic suffix of $w$ longer than $(QyxPxy)^{n}Q$ must start in $u$, in order to prove the claim it is enough to show that the first non-prefix occurrence of $u$ in $w$ is that appearing as a prefix of $(QyxPxy)^{n}Q$. Now, since the prefix $v=PxyQyxP$ of $w$ can be written as $v=uyxP=Pxyu$, one has by Proposition \ref{prop:charcentral} that $v$ is a central word. It is easy to prove (see, for example, \cite{BuDelFi12}) that the longest palindromic suffix of a central word does not have internal occurrences, that is, appears in the central word only as a prefix and as a suffix. Therefore, since $|u|>|P|$, $u$ is the longest palindromic suffix of $v$ (by Proposition \ref{prop:charcentral}), and so appears in $v$ only as a prefix and as a suffix. This shows that $(QyxPxy)^{n}Q$ is the longest palindromic suffix of $w$.

Thus, we have $w^{(+)}=wyxP$, and we can write: 
\begin{eqnarray*}
 w^{(+)} & = & Pxy(QyxPxy)^{n}QyxP\\
 & = & PxyQ \cdot yx \cdot P(xyQyxP)^{n}\\
 & = & (PxyQyx)^{n}P \cdot xy \cdot QyxP
\end{eqnarray*}
so that $w^{(+)}=uyxz=zxyu$ for the palindrome $z=P(xyQyxP)^{n}=(PxyQyx)^{n}P$. By Proposition \ref{prop:charcentral}, $w^{(+)}$ is a central word.\qed
\end{proof}

We are now ready to state our main result.

\begin{theorem}\label{theor:main}
$\BS=\{w\mid xwy \mbox{ is a Christoffel word, $x,y\in \Sigma\}$}.$
\end{theorem}

\begin{proof}
 Let $xwy$ be a Christoffel word, $x,y\in \Sigma$.  Then, by Lemma \ref{lem:npChris}, $w$ is of the form $w=(uyx)^{n}u$, $n\ge 0$, for a central word $u$. By Lemma \ref{lem:rpl}, $w$ is a prefix of the central word $w^{(+)}$ and a suffix of the central word $w^{(-)}$, and therefore, by Proposition \ref{prop:sbscen} and Lemma \ref{lem:prefsuf}, $w$ is a bispecial Sturmian word.
 
 Conversely, let $w$ be a bispecial Sturmian word, that is, suppose that the words $xw$, $yw$, $wx$ and $wy$ are all Sturmian. If $w$ is strictly bispecial, then $w$ is a central word by Proposition \ref{prop:sbscen}, and $xwy$ is a (primitive) Christoffel word by Theorem \ref{theor:sbsCP}. So suppose $w\in \WBS$. By Lemma \ref{lem:npChris}, it is enough to prove that $w$ is of the form $w=(uyx)^{n}u$, $n\ge 1$, for a central word $u$ and letters $x\neq y$. Since $w$ is not a strictly bispecial Sturmian word, it is not a palindrome (by Proposition \ref{prop:sturmstrispe}). Let $u$ be the longest palindromic border of $w$ (that is, the longest palindromic prefix of $w$ that is also a suffix of $w$), so that $w=uyzxu$, $x\neq y\in \Sigma$, $z\in \Sigma^{*}$. If $z=\varepsilon$, $w=uyxu$ and we are done. Otherwise, it must be $z=xz'y$ for some $z'\in \Sigma^{*}$, since otherwise either the word $yw$ would contain $yuy$ and $xxu$ as factors (a contradiction with the hypothesis that $yw$ is a Sturmian word) or the word $wx$ would contain $uyy$ and $xux$ as factors (a contradiction with the hypothesis that $wx$ is a Sturmian word). 
 
So $w=uyxz'yxu$. If $u=\varepsilon$, then it must be $z=(yx)^{k}$ for some $k\ge 0$, since otherwise either  $xx$ would appear as a factor in $w$, and therefore the word $yw$ would contain $xx$ and $yy$ as factors, being not a Sturmian word, or $yy$ would appear as a factor in $w$, and therefore the word $wx$ would contain $xx$ and $yy$ as factors, being not a Sturmian word. Hence, if $u=\epsilon$ we are done, and so we suppose $|u|>0$. 
 
 By contradiction, suppose that $w$ is not of the form $w=(uyx)^{n}u$. That is, let $w=(uyx)^{k}u'av$, with $k\ge 1$, $v\in \Sigma^{*}$, $u'b\in \Pref(uyx)$, for different letters $a$ and $b$. If $|u'|\ge |u|$, then either $|u'|=|u|$ or $|u'|=|u|+1$. In the first case, $u'=u$ and $w=(uyx)^{k}uxv'$, for some $v'\in \Sigma^{*}$, and then the word $yw$ would contain $yuy$ and $xux$ as factors, being not a Sturmian word. In the second case, $u'=uy$ and $w=(uyx)^{k}uyyv''$, for some $v''\in \Sigma^{*}$; since $xu$ is a suffix of $w$, and therefore $w=(uyx)^{k}v'''xu$ for some $v'''\in \Sigma^{*}$, we would have that the word $wx$ contains both $uyy$ and $xux$ as factors, being not a Sturmian word. Thus, we can suppose $u'b\in \Pref(u)$. Now, if $a=x$ and $b=y$, then the word $yw$ would contain the factors $yu'y$ and $xu'x$, being not a Sturmian word; if instead $a=y$ and $b=x$, let $u=u'xu''$, so that we can write $w=(uyx)^{k}u'yv=(uyx)^{k-1}u'xu''yxu'yv$. The word $wx$ would therefore contain the factors $u''yxu'y$ and $xux=xu'xu''x$ (since $xu$ is a suffix of $w$), being not a Sturmian word (see~\figurename~\ref{fig:theorem}). In all the cases we obtain a contradiction and the proof is thus complete.\qed
\end{proof}

\begin{figure}[ht]
\begin{center}
\includegraphics[height=32mm]{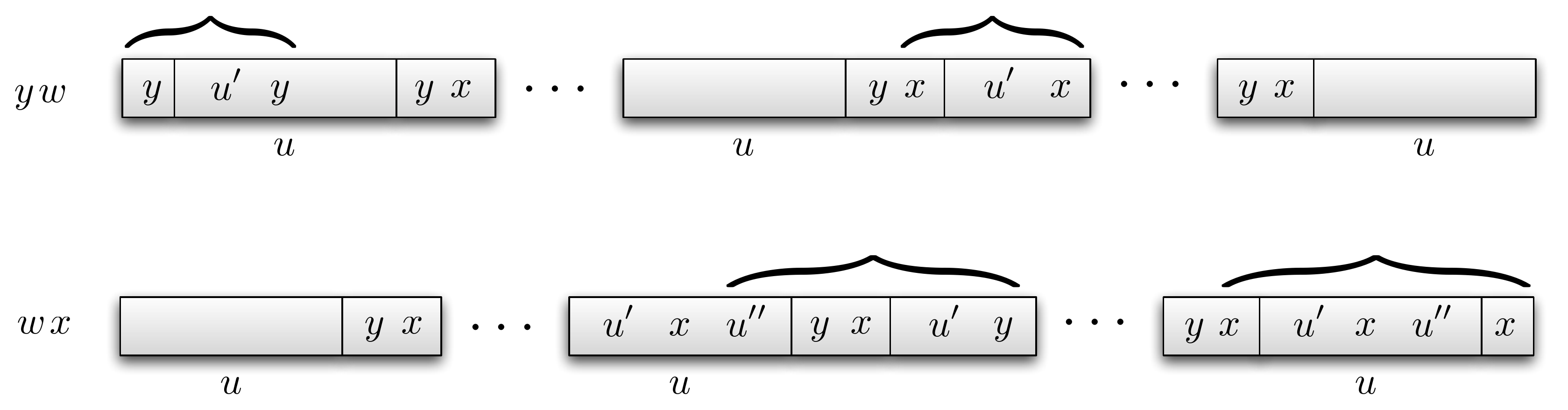}
\caption{The proof of Theorem \ref{theor:main}.}
\label{fig:theorem}
\end{center}
\end{figure}

So, bispecial Sturmian words are the maximal internal factors of Christoffel words. Every bispecial Sturmian word is therefore of the form $w=(uyx)^{n}u$, $n\ge 0$, for different letters $x,y$ and a central word $u$. The word $w$ is strictly bispecial if and only if $n=0$. If $n=1$, $w$ is a \emph{semicentral word} \cite{BuDelFi12}, that is, a word in which the longest repeated prefix, the longest repeated suffix, the longest left special factor and the longest right special factor all coincide.

\section{Enumeration of bispecial Sturmian words}\label{sec:En}

In this section we give an enumerative formula for bispecial Sturmian words. It is known that the number of Sturmian words of length $n$ is given by
\[St(n)=1+\sum_{i=1}^{n}(n-i+1)\phi(i)\]
where $\phi$ is the Euler totient function, i.e., $\phi(n)$ is the number of positive integers smaller than or equal to $n$ and coprime with $n$ (cf.~\cite{Mig91,Lip82}).

Let $w$ be a Sturmian word of length $n$. If $w$ is left special, then $aw$ and $bw$ are Sturmian words of length $n+1$. If instead $w$ if not left special, then only one between $aw$ and $bw$ is a Sturmian word of length $n+1$. Therefore, we have $\LS(n)=St(n+1)-St(n)$ and hence
\[\LS(n)=\sum_{i=1}^{n+1}\phi(i)\]

Using a symmetric argument, one has that also \[\RS(n)=\sum_{i=1}^{n+1}\phi(i)\]

Since \cite{DelMi94} $\SBS(n)=\LS(n+1)-\LS(n)=\RS(n+1)-\RS(n)$, we have
\[\SBS(n)=\phi(n+2)\]

Therefore, in order to find an enumerative formula for bispecial Sturmian words, we only have to enumerate the non-strictly bispecial Sturmian words. We do this in the next proposition.

\begin{proposition}
For every $n>1$, one has \[\WBS(n)=2\left(n+1-\phi(n+2)\right)\]
\end{proposition}

\begin{proof}
Let \[W_{n}=\{w \mid \mbox{  $awb$ is a lower Christoffel word of length $n+2$} \}\] and \[W'_{n}=\{w' \mid \mbox{  $bw'a$ is an upper Christoffel word of length $n+2$} \}\] By Theorem \ref{theor:main}, the bispecial Sturmian words of length $n$ are the words in $W_{n}\cup W'_{n}$. 

Among the $n+1$ words in $W_{n}$, there are $\phi(n+2)$ strictly bispecial Sturmian words, that are precisely the palindromes in $W_{n}$. The $n+1-\phi(n+2)$ words in $W_{n}$ that are not palindromes are non-strictly bispecial Sturmian words. The other non-strictly bispecial Sturmian words of length $n$ are the $n+1-\phi(n+2)$ words in $W'_{n}$ that are not palindromes. Since the words in $W'_{n}$ are the reversals of the words in $W_{n}$, and since no non-strictly bispecial Sturmian word is a palindrome by Proposition \ref{prop:sturmstrispe}, there are a total of $2(n+1-\phi(n+2))$ non-strictly bispecial Sturmian words of length $n$.\qed
\end{proof}

\begin{corollary}\label{cor:formula}
 For every $n\ge 0$, there are $2(n+1)-\phi(n+2)$ bispecial Sturmian words of length $n$.
\end{corollary}

\begin{example}
The Christoffel words of length $12$ and their maximal internal factors, the bispecial Sturmian words of length $10$, are reported in Table \ref{tab:example}. 

\begin{table}[h]
\begin{center}
  \begin{tabular}{|c | c | c | }
  
\ pair $(p,q)$ \ & \ lower Christoffel word $w_{p,q}$ \ & \ upper Christoffel word $w'_{p,q}$ \  \\    \hline 
 $(11,1)$ &    $a\underline{aaaaaaaaaa}b$   & $b\underline{aaaaaaaaaa}a$        \\
 $(10,2)$ &    $aaaaabaaaaab$   & $baaaaabaaaaa  $      \\
 $(9,3)$  &    $aaabaaabaaab$   & $baaabaaabaaa $       \\
 $(8,4)$  &    $aabaabaabaab$   & $baabaabaabaa $       \\
 $(7,5)$  &    $a\underline{ababaababa}b$   & $b\underline{ababaababa}a $       \\
 $(6,6)$  &    $abababababab$   & $babababababa $    \\
 $(5,7)$  &    $a\underline{bababbabab}b$   & $b\underline{bababbabab}a $       \\
 $(4,8)$  &    $abbabbabbabb$   & $bbabbabbabba  $      \\
 $(3,9)$  &    $abbbabbbabbb$   & $bbbabbbabbba $       \\
 $(2,10)$ &    $abbbbbabbbbb$   & $bbbbbabbbbba $       \\
 $(1,11)$ &    $a\underline{bbbbbbbbbb}b$   & $b\underline{bbbbbbbbbb}a $       \\    
 \hline 
  \end{tabular}\vspace{4mm}
\end{center}\caption{The Christoffel words of length $12$. Their maximal internal factors are the bispecial Sturmian words of length $10$. There are $4=\phi(12)$ strictly bispecial Sturmian words, that are the palindromes $aaaaaaaaaa$, $ababaababa$, $bababbabab$ and $bbbbbbbbbb$ (underlined), and $14=2(11-4)$ non-strictly bispecial Sturmian words: $aaaaabaaaa$, $aaaabaaaaa$, $aaabaaabaa$, $aabaaabaaa$, $aabaabaaba$, $abaabaabaa$, $ababababab$, $bababababa$, $babbabbabb$, $bbabbabbab$, $bbabbbabbb$, $bbbabbbabb$, $bbbbabbbbb$ and  $bbbbbabbbb$.\label{tab:example}}
\end{table}
\end{example}

\section{Minimal forbidden words}\label{sec:MF}

Given a factorial language $L$ (that is, a language containing all the factors of its words) over the alphabet $\Sigma$, a word $v\in \Sigma^{*}$ is a minimal forbidden word for $L$ if $v$ does not belong to $L$ but every proper factor of $v$ does (see~\cite{CrMiRe98} for further details). Minimal forbidden words represent a powerful tool to investigate the structure of a factorial language (cf.~\cite{BeMiRe96}). In the next theorem, we give a characterization of the set $\MF$ of minimal forbidden words for the language $\St$.

\begin{theorem}\label{theor:mf}
 $\MF=\{ywx \mid xwy \mbox{ is a non-primitive Christoffel word, $x,y\in \Sigma$}\}.$
\end{theorem}

\begin{proof}
If $xwy$ is a non-primitive Christoffel word, then by Theorems \ref{theor:sbsCP} and  \ref{theor:main}, $w$ is a non-strictly bispecial Sturmian word. This implies that $ywx$ is not a Sturmian word, since a word $w$ such that $xwy$ and $ywx$ are both Sturmian is a central word \cite{DelMi94}, and therefore a strictly bispecial Sturmian word (Proposition \ref{prop:sbscen}). Since $yw$ and $wx$ are Sturmian words, we have $ywx\in \MF$.

Conversely, let $ywx\in \MF$. By definition, $yw$ is Sturmian, and therefore $ywy$ must be a  Sturmian word since $\St$ is an extendible language. Analogously, since $wx$ is Sturmian, the word $xwx$ must be a Sturmian word. Thus, $w$ is a bispecial Sturmian word, and since $ywx\notin \St$, $w$ is a  non-strictly bispecial Sturmian word. By Theorems \ref{theor:sbsCP} and \ref{theor:main}, $xwy$ is a non-primitive Christoffel word.\qed
\end{proof}

\begin{corollary}\label{cor:formulamf}
 For every $n>1$, one has \[\MF(n)=2(n-1-\phi(n))\]
\end{corollary}

It is known from \cite{Mig91} that $\St(n)=O(n^{3})$, as a consequence of the estimation (see~\cite{HaWr}, p. 268)
\begin{equation}\label{eq:HW}
 \sum_{i=1}^{n}\phi(i)=\frac{3n^{2}}{\pi^{2}}+O(n\log n)
\end{equation}
From (\ref{eq:HW}) and from the formula of Corollary \ref{cor:formulamf}, we have that $$\sum_{i=1}^{n}\MF(n)=O(n^{2})$$

\bibliographystyle{abbrv}
\bibliography{bisturmian}
\end{document}